\documentclass[draftcls, 11pt, onecolumn]{IEEEtran}
\usepackage{amsmath,amsthm,amssymb,paralist,subfigure,cite,graphicx,color}
\usepackage{algorithm2e}
\newtheorem{lem}{Lemma}
\newtheorem{theorem}{Theorem}
\newtheorem{defn}{Definition}

\newtheorem{rem}{Remark}
\newtheorem{ass}{Assumption}

\def\mc{\mathcal}

\begin{document}

\title{Finite-Time Stability Under Denial of Service}

\author{Mohammadreza Doostmohammadian, and Nader Meskin, \textit{Senior member, IEEE}
	\thanks{Mohammadreza Doostmohammadian is with the Faculty of Mechanical Engineering at Semnan University, Semnan, Iran, email: \texttt{doost@semnan.ac.ir}. Nader Meskin is with the Electrical Engineering Department at Qatar University, Doha, Qatar, email: \texttt{nader.meskin@qu.edu.qa}.}
\thanks{This paper was supported  by the Qatar National Research Fund (a member of the Qatar Foundation) under the NPRP Grant number NPRP10-0105-170107.}	
	}
\maketitle

\begin{abstract}
	Finite-time stability of networked control systems under Denial of Service (DoS) attacks are investigated in this paper, where the communication between the plant and the controller is compromised at some time intervals.  Toward this goal,  first an event-triggered mechanism based on the variation rate of the  Lyapunov function is proposed such that  the closed-loop system remains finite-time  stable (FTS) and at the same time, the amount data exchange in the network is reduced. Next, the vulnerability of the proposed event-triggered finite-time controller in the presence of DoS attacks are evaluated  and sufficient conditions on the DoS duration and frequency are obtained to assure the finite-time stability of the closed-loop system in the presence of DoS attack where no assumption on  the DoS attack in terms of following a certain probabilistic  or a well-structured periodic model is considered. Finally, the efficiency of the proposed approach is demonstrated through a simulation study.

	\textit{Keywords:} Finite-time convergence, Input-to-state stability, Homogeneous systems, Denial-of-Service, Networked control systems, Event-triggered control
	
\end{abstract}

\section{Introduction}\label{sec_introduction}
\IEEEPARstart{I}{ndustrial}
control systems are typically equipped with information-sharing and communication  facilities for the transmission of measurement and control data. However, wireless networks and the Internet, as the key components of such control systems, are prone to disruption and cyber attacks. This is more challenging in  large-scale networked control systems and, particularly, in recent emerging context of Cyber-Physical-Systems (CPS) \cite{isj_cyber} and Internet-of-Things (IoT). In this direction, cyber-security is a recent topic of interest in the literature \cite{wood2002denial,xu2006jamming,asilomar14,khaitan2014design,chen2016dynamic,doostmohammadian2017distributed,cetinkaya2019overview,yadegar2019output,cardenas2008secure}, where different attack detection/mitigation strategies along with resilient control approaches are evaluated to  guarantee a safe and secure operation of the closed-loop system despite the presence
of failure, disruption of service, or possible malicious attacks.

A Denial-of-Service attack (DoS attack) refers to a type of cyber-attack in which the attacker aims to make  network resources unavailable to users by disrupting the services of a host connected to the Internet or the network. DoS is typically made by flooding the targeted channel or resource with redundant requests to overload the system and prevent  authorized requests to be accomplished.
Similarly, in a Distributed Denial-of-Service attack (DDoS attack), the incoming traffic flooding the targeted resource comes from many different sources, including targeted online password guessing \cite{wang2016targeted}. This effectively makes it impossible to stop the attack simply by blocking a single source. In networked control systems, DoS and DDoS have two similar consequences on the data transmission between the controller and the plant, namely  long delay jitters and large amount of packet loss \cite{BEITOLLAHI2011107} and the only difference between DoS and DDoS is on how they are deployed by the attacker where in DoS attack only one computer and one internet connection is used to flood the targeted system while in DDoS multiple computers and internet connections are utilized for attacking the targeted system. However, in terms of their ultimate effects on networked control systems, they have similar behavior.

\textit{Motivation:} According to \cite{cardenas2008secure}, security of  CPS is a critical issue which differs from general computing systems in the sense that any cyber-attack including DoS may cause disruption in the underlying physical system.  For example, in this paper, we consider sampled-data control systems in which  the plant-controller communication channel is subject to DoS (or DDoS). This may cause instability in the closed-loop systems due to denied communication on  the control input  channel (controller to actuator) and the measurement channel (sensor to controller), and in turn, it may result in critical damages to the physical system. While some works  assume the DoS attack  follows a probability distribution \cite{schenato2007foundations}, here we are concerned with the deterministic conditions under which the closed-loop finite-time stability is preserved.

The main motivation of this paper is to investigate the vulnerability of finite-time stability of networked control systems and to obtain sufficient conditions on the frequency and duration of DoS/DDoS attack such that the closed-loop networked control system remains finite-time stable. It should be noted that the proposed approach/framework is not considered as a mitigation solution and DoS/DDoS mitigation is mainly addressed by designing a secure architecture to 	protect a given networked control system from DoS/DDoS attack which is mainly a computer network problem handled by IT/computer engineering experts.  However, in this paper, we are considering the problem from the control engineering perspective and the worst case scenario is analyzed in which the attacker has successfully  launched a DoS/DDoS attack on a networked control system and we mainly investigate the attack effects on the finite-time stability of the closed-loop system. 

\textit{Literature review:} DoS attacks within the framework of \textit{linear} systems under state-feedback is considered in \cite{de2015input} where  a sampling strategy is proposed to ensure the \textit{exponential} input-to-state stability in the presence of DoS. Their approach adopts \textit{event-triggered} sampling methods, as in \cite{tabuada2007event,abdelrahim2015stabilization}, that properly constrain the closed-loop trajectories to assure the closed-loop stability. Such event-triggered control scenarios are prevalent in literature, e.g. see \cite{dolk2016event,foroush2012event,de2016networked,shisheh2016triggering,chen2018event,girard2014dynamic,al2019improved}, to account for limited network resources such as limited bandwidth in wireless networks.  In \cite{dolk2016event},  output-based resilient design conditions are developed such that the resulting closed-loop nonlinear Lipschitz  system is input-to-output stable, where in \cite{de2016networked} global exponential stability of networked control systems under DoS is considered.
Triggering control techniques to ensure asymptotic stability of linear systems under  well-structured \textit{periodic} \cite{foroush2012event}, energy-constrained \cite{shisheh2016triggering}, and \textit{Pulse-Width Modulated} (PWM) DoS attacks \cite{chen2018event} are also addressed in the literature.
Game-theoretic approaches assuming an intelligent jammer are also considered in \cite{gupta2010optimal,gupta2012dynamic} where in \cite{gupta2010optimal} a threshold-strategy is addressed, while the jamming attack occurs whenever the system state is larger than a specific threshold. In \cite{gupta2012dynamic}, an energy-constrained jamming scenario based on the full knowledge of the system state is considered.


\textit{Contribution:} This paper aims to characterize the duration and frequency of the DoS attack under  which the closed-loop system remains finite-time stable. Unlike \cite{cetinkaya2018probabilistic,cetinkaya2018analysis} in which  a probabilistic packet drop model is considered for the DoS attack and similar to  \cite{de2016networked,de2015input,dolk2016event}, no restricting assumption on the attack strategy is considered here. The main contribution of this work is to relate the finite-time stability properties  to the duration and frequency of DoS on/off transitions. As compared to asymptotic stability of linear \cite{de2015input,foroush2012event,shisheh2016triggering,chen2018event} and Lipschitz nonlinear \cite{de2016networked,dolk2016event} systems, for finite-time stability the underlying nonlinear system needs to be non-Lipschitz at the  origin (or equilibrium point) \cite{bhat2000finite,scientia,single_bit}. Such finite-time protocols, initially introduced in optimal control literature \cite{ryan1982optimal}, are of interest due to  reducing the response time \cite{haimo1986finite} and forcing the system to reach the desired target in finite-time \cite{hong2010finite}.

In this paper, using the results governing the finite-time input-to-state stable (FTISS) systems and adopting an event-triggered mechanism that suitably constrains the system trajectories, a Lyapunov-based analysis is developed to assure the finite-time stability under DoS attack. Particularly, we derive the  bound on DoS on/off transitions such that the stability during the \textit{off-periods} of DoS dominates the instability during the \textit{on-periods} of DoS, where during the on-periods of DoS the open-loop system evolves using the most recent transmitted control signal. There is no constraint on the control input during the off-periods of DoS and any type of state-feedback control design, e.g. robust control, can be considered. Note that the event-triggered mechanism is designed in accordance with the variation rate of the Lyapunov function and managing the sampling rate accordingly ensures the resiliency of our method, and further, allows for sufficiently flexible design to account for communication resources. Finally, the performance of the proposed approach is demonstrated and compared with a relevant work in the literature through a simulation study.

\section{The Framework} \label{sec_prob}
\subsection{Finite-time Input-to-State Stability}
Consider a nonlinear system in the form:
\begin{eqnarray} \label{eq_syst}
\dot{x}(t) =  f(x(t), u(t)),~x(0)=x_0,
\end{eqnarray}
where  $x(t) \in \mathbb{R}^n$, and $u(t) \in \mathbb{R}^m$ represent the system state and the system input, respectively. Considering a closed-loop feedback controller $u(t) = \psi(x(t),e(t))$ where $e(t)$ represent the closed-loop error signal due sampling,  one can rewrite the closed-loop system as:
\begin{eqnarray} \label{eq_syst2}
\dot{x}(t) =  f(x(t), {\psi}(x(t),e(t)))=F(x(t),e(t)),~x(0)=x_0.
\end{eqnarray}
First, we formally define the finite-time stability  and finite-time input-to-state stability  properties of the closed-loop system and the required conditions in terms of Lyapunov stability.

\begin{defn}
	The equilibrium $x=0$ of system \eqref{eq_syst2} with $e(t)=0$ is finite-time stable (locally) if it is locally stable in the Lyapunov sense and
	for any initial time $t_0$ and initial state $x(t_0)=x_0 \in \mathcal{V}$ with $\mathcal{V}$ as a nonempty neighborhood of the origin in $\mathbb{R}^n$, there exists a settling-time function  $T(x_0)$  such that
	\begin{eqnarray} \nonumber
     T(x_0)=\inf\{T>0:~\lim\limits_{t\rightarrow T} x(t,x_0)=0,\\  x(t,x_0) = 0 ~ \forall t>T\}
     \label{eq_FTS}
	\end{eqnarray}	
\end{defn}
\noindent where $x(t,x_0)$ denotes the state trajectory of system \eqref{eq_syst2} with the initial condition $x_0$. 	It should be noted  that for an autonomous system $\dot{x}(t)=F(x(t))$ to be FTS, it is necessary that the function $F(x)$ be \textit{non-Lipschitz} at the equilibrium point $x=0$ \cite{bhat2000finite,scientia}. Example of non-Lipschitz functions are  $\mbox{sgn}(x)$ or $\mbox{sgn}(x)|x|^a$, $0<a<1$.

\begin{lem} \label{lem_finite} \cite{bhat2000finite}
	Consider the autonomous system $\dot{x}(t)=F(x(t))$.  Assume that there exist a continuously differentiable function $V$, real numbers $c>0$ and $a \in (0,1)$ such that $V(0)=0$, $V(x)>0, x \in \mathcal{V}\subset \mathbb{R}^n$  and
	\begin{align}
	\dot{V}(x(t))=\nabla V(x(t))F(x(t))\leq -c(V(x(t))^a,~~~x \in \mathcal{V}.	\label{eq_V_FTS}
	\end{align}
Then, the origin is  a finite-time stable equilibrium of the system. Moreover, the settling-time function $T(x_0)$ is given as:
	\begin{align}\nonumber
	T(x_0) &\leq \frac{1}{c(1-a)}(V(x_0))^{1-a}.
	\end{align}
\end{lem}

\begin{defn} \cite{hong2010finite}
	The closed-loop system \eqref{eq_syst2} is locally finite-time input-to-state stable with respect to the input signal $e(t)$  if  for every $x_0 \in  \mathcal{V} \subset \mathbb{R}^n$, and every bounded input $e(t)\in \mathbb{R}^n$ with $\|e\|_\infty< \rho$, we have
	\begin{align} \label{eq_FTISS}
	\|x(t)\|\leq \beta(\|x_0\|,t)+\bar{\gamma} \big (\sup_{0\leq\tau\leq t}  \|e(\tau)\| \big)
	\end{align}	
    where $\bar{\gamma}$ is a class $\mc{K}_\infty$-function\footnote{A function $\gamma:\mathbb{R}_{\geq 0} \rightarrow \mathbb{R}_{\geq 0} $ is of class $\mc{K}$-function  if  it is continuous,  strictly increasing, and $\gamma(0)=0$. Further, it is of class  $\mc{K}_\infty$-function if it is also unbounded ($\gamma(r) \rightarrow \infty$ as $r \rightarrow \infty$).} and $\beta$ is a generalized $\mc{K}\mc{L}$-function\footnote{A function $\beta:\mathbb{R}_{\geq 0} \times \mathbb{R}_{\geq 0} \rightarrow \mathbb{R}_{\geq 0}$ is of class generalized $\mc{K}\mc{L}$-function if $\beta(.,t)$ is of class  $\mc{K}_\infty$-function for all $t$ and $\beta(r,t) \rightarrow 0$ as $t \rightarrow T$ for some $T<\infty$.} with $\beta(\|x_0\|,t)=0$ when $t \geq T$ with $T$ as a continuous function of $x_0$.
\end{defn}

It should be noted that  the main difference  between ISS and FTISS is that  for an ISS system, $\beta(\|x_0\|,t) \rightarrow 0$ as $t \rightarrow \infty$, while for FTISS system,  $\beta(\|x_0\|,t)=0$ as $t \geq T$ with $T<\infty$.

\begin{defn} \cite{agrachev2008nonlinear} \label{Def3}
   A function $V:\mathcal{V}\rightarrow \mathbb{R}_{\geq0}$ is called an \textit{ISS-Lyapunov} function for system \eqref{eq_syst2}, if  there exist class $\mc{K}_\infty$ functions  $\alpha_1,\alpha_2, \alpha_3, {\gamma}$ such that,
	\begin{eqnarray} \label{eq_storage}
	\alpha_1(\|x\|) &\leq& V(x) \leq \alpha_2(\|x\|), \\
    \nonumber
	\dot{V}(x(t),e(t))&\leq& -\alpha_3(\|x(t)\|)+{\gamma}(\|e(t)\|).  \label{eq_VV}
	\end{eqnarray}		
for all $x \in \mathcal{V}$.
\end{defn}

Defining $-\alpha_3(\|x(t)\|) \leq -\alpha_3(\alpha_1^{-1}(V(x(t))))$ and $ \alpha_3(\alpha_1^{-1}(x(t)))=\alpha(x(t))$, then condition \eqref{eq_VV} can be restated as,
	\begin{eqnarray} \nonumber
	\dot{V}(x(t),e(t))&=&\nabla V(x(t)).F(x(t),e(t)) \\ &\leq& -\alpha(V(x(t)))+{\gamma}(\|e(t)\|).
	\label{eq_V2}
	\end{eqnarray}
	with ${\alpha}  \in \mc{K}_\infty$.	

\begin{lem}  \label{lem_lyapunov_order}
	\cite{hong2008finite}
	System~\eqref{eq_syst2} is FTISS if there exists an ISS-Lyapunov function such that in \eqref{eq_V2}, ${\alpha}(V)=\mc{O}(V^a)$ with $0<a<1$, i.e.
	\begin{eqnarray} \label{eq_V3}
	\dot{V}(x(t),e(t)) \leq -cV^a(x(t))+{\gamma}(\|e(t)\|), ~~~\forall x \in \mathcal{V},
	\end{eqnarray}
	with $c>0$, and ${\gamma} \in \mc{K}_\infty$.
\end{lem}

\subsection{Background on FTS and FTISS Systems}
Finite-time stable (FTS) systems find application where instead of typical asymptotic convergence, the finite-time convergence is desired, see some applications in \cite{single_bit}.
Such models are particularly associated with  \textit{non-smooth} feedback laws to stabilize systems which, otherwise, are unstabilizable by smooth feedback.
In this direction, system \textit{homogeneity} is a related concept and it is known that any stable homogeneous system with \textit{negative} homogeneity degree is FTS \cite{polyakov2019consistent,bernuau2015robust}. For example, in \cite{anta2010sample}
input-to-state stability of homogeneous systems is studied for triggering control of nonlinear systems. Homogeneous systems  find applications in, e.g. sliding mode control \cite{levant2005homogeneity} and fixed-time stability  \cite{bernuau2015robust} among others. Finite-time stability was first studied in the
optimal control literature \cite{ryan1982optimal} and finite-time
controllers generally result in a fast response and high tracking precision as well as disturbance-rejection properties due to their non-smoothness \cite{hong2010finite}.
Similarly, FTISS systems have the same privileges over the ISS systems, including the finite-time convergence among others.
For better understanding of the ISS and FTISS systems and their differences we refer interested readers to \cite{wang2013finite,hong2010finite,agrachev2008nonlinear,hong2008finite,sontag2008input}. Note that FTISS systems are comparatively  less studied in comparison
with smooth systems and the literature is limited to the mentioned references.

\subsection{Control Objective}
In this paper, we assume that the closed-loop system, in an ideal continuous-time case, is FTISS. In this direction the following assumption is made:
\begin{ass} \label{ass_lyapunov}
	For system~\eqref{eq_syst2}, there exists an ISS-Lyapunov function $V:\mathcal{V}\rightarrow \mathbb{R}_{\geq0}$, which satisfies Lemma \ref{lem_lyapunov_order} for  given class $\mc{K}_\infty$ functions  $\alpha,\gamma$, and  constants $0<a<1$, $c>0$.
\end{ass}

\begin{figure}
	\centering
	\includegraphics[width=3.3in]{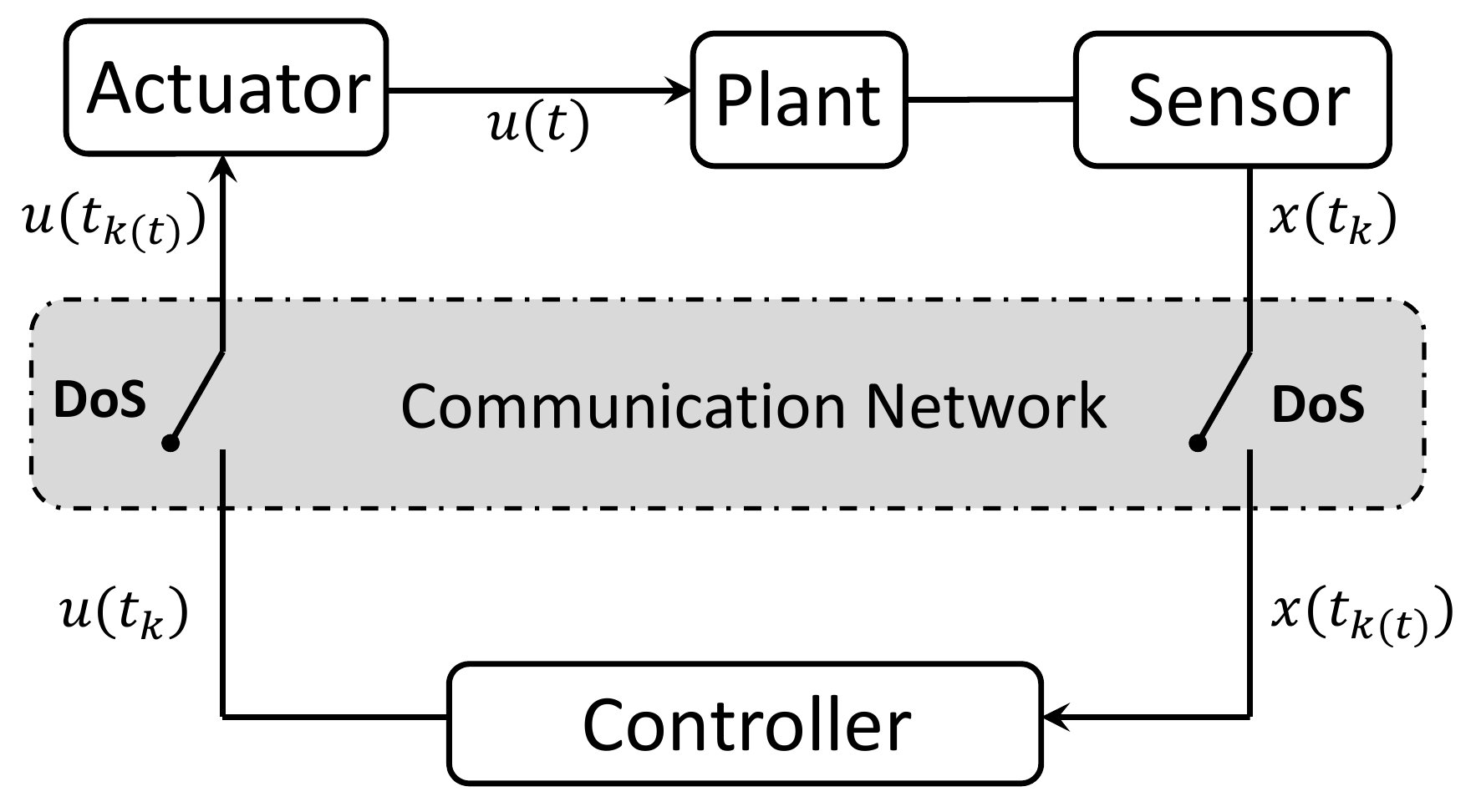}
	\caption{The block-diagram of the closed-loop system under DoS attack.  }
	\label{fig_model}
\end{figure}

The block diagram of the considered networked control system is shown in Fig.\ref{fig_model} where  the information exchange between sensor/actuator and controller is done through communication channels. In this paper, it is assumed that the attacker can compromise the security of the communication channels to inject  DoS/DDoS attack.
In the presence of DoS/DDoS, the input $u(t)$ is generated based on the most recently received signal when no DoS was present.
Then, the problem is under what conditions on DoS attack the closed-loop finite-time stability is preserved. In this paper, we consider  justified assumptions on the duration/frequency of the DoS signal along with  assumption to prevent finite escape time in the system. Under these assumptions, borrowing ideas from event-based
control and FTISS Lyapunov functions, the problem of interest is to develop an event-triggered control such that in the presence of DoS attack, the closed-loop system remains FTS and the inter-execution times of the controller are bounded away from zero to avoid Zeno phenomena\footnote{The Zeno phenomenon (or Zeno behavior) refers to the phenomenon of infinite number of events over a finite-time period. In this paper, and in general control literature, the Zeno phenomenon implies infinite number of control updates and sampling over finite-time.}. We particularly derive the conditions on DoS attack such that, using modified hybrid event-triggered mechanism, the system remains FTS under DoS. 
It should be mentioned that we assume all the entities in the considered networked control system shown in Fig. \ref{fig_model}  (i.e. plant, controller, sensor, and actuator) perform their specified actions using some sort of suitable user authentication schemes such as  the ones given in \cite{wang2016two,wang2016challenges}. However, it is assumed that the attacker can still compromise the implemented authentication scheme and launch DoS/DDoS attack on the networked control system shown in Fig.  \ref{fig_model} .

We summarize the  problems in this paper as follows:

\textbf{Problem 1:} Development of a Zeno-free event-triggered mechanism such that the closed-loop system remains FTS without considering the DoS attack.

\textbf{Problem 2:} Extending the event-triggered mechanism obtained in the first problem in the presence of DoS attack and deriving an upper-bound relation on the frequency and duration of DoS intervals such that the closed-loop system remains FTS.

\section{Event-triggered Mechanism} \label{sec_sample}
In this section, Problem 1 is considered and an event-triggered mechanism  is developed such that the closed-loop system remains FTS without considering the DoS attack. It is  also shown that the Zeno phenomena is almost excluded. To propose our sampling scheme, we make the following assumption in the paper.
\begin{ass} \label{ass_mu}
	For  $x \in \mathcal{V}$, there exists $\mu>0$ such that ${\gamma}(4\|x\|)\leq\mu \alpha_1^a(\|x\|)$ where the functions $\alpha_1$ and $\gamma$ are defined in Definition \ref{Def3} and  $\mathcal{V}$ is a nonempty neighborhood of the origin in $\mathbb{R}^n$.
\end{ass}
In the proposed event-triggered framework, the event instants, denoted by $t_k$, $k=0,1,\dots$,  are generated based on the following event-triggered mechanism:
	\begin{eqnarray} \label{eq_sampling1}
	t_{k+1} = \inf\{t>t_k~|~{\gamma}(4\|e(t)\|)>c(1-\lambda)V^a(x(t))\},
	\end{eqnarray}
with $0<\lambda<1$, $e(t) = x(t_{k})-x(t)$ and the Lyapunov function $V(x(t))$ satisfies the FTISS condition in \eqref{eq_V3}.
\begin{lem} \label{lem_zeno1}
	The event-triggered mechanism~\eqref{eq_sampling1} is almost always Zeno-free.
\end{lem}
\begin{proof}
	Note that for $t_k<t<t_{k+1}$ we have $\gamma(4\|e(t)\|)\leq c(1-\lambda)V^a(x(t))$.
	It is known that any event-triggered mechanism for which the error is restricted to satisfy $\gamma(\|e\|)\leq K\sigma(x)$ is Zeno-free wherever $\gamma$ and $\sigma^{-1}$ functions are Lipschitz \cite{tabuada2007event}.  Note that the Lipschitz continuity is a sufficient condition for Zeno-freeness not a necessary condition. Since the functions $\gamma$ and $V^{-a}$ are Lipschitz \textit{almost} everywhere, the event-triggered mechanism  \eqref{eq_sampling1} does not show  Zeno behavior almost everywhere.
\end{proof}
Similar analysis as in the above proof is given in
\cite{du2018finite} to prove the \textit{almost} Zeno-freeness of an event-triggered scheme.

\begin{theorem}
	The system~\eqref{eq_syst} with sample-and-hold input $u(t)=\psi(x(t_k))$, $t\in [t_k,t_{k+1})$ under the event-triggering mechanism~\eqref{eq_sampling1} is FTS. Particularly, the Lyapunov function is such that,
	\begin{eqnarray} \label{eq_V_DoS1}
	V(x(t))^{1-a} \leq V(x(t_k))^{1-a}-(1-a)\omega_1(t-t_k)
	\end{eqnarray}
	where $\omega_1= c\lambda$ and $t_k<t<t_{k+1}$.
\end{theorem}
\begin{proof}
	Following Assumption~\ref{ass_lyapunov}, we have an ISS-Lyapunov function such that $\dot{V}(x(t),e(t)) \leq -cV^a(x(t))+{\gamma}(\|e(t)\|)$. Further, the event-triggering mechanism~\eqref{eq_sampling1} implies that for $t_k<t<t_{k+1}$ we have $\gamma(4\|e(t)\|)\leq c(1-\lambda)V^a(x(t))$. Since ${\gamma} \in \mc{K}_\infty$, we have $\gamma(\|e(t)\|)<\gamma(4\|e(t)\|)$ and  it follows that,
	\begin{align} \nonumber
	\dot{V}(x(t)) \leq& -cV^a(x(t))+c(1-\lambda)V^a(x(t))\\ \label{eq_vdot_omega1}
	\leq& -c\lambda V^a(x(t)),
	\end{align}
	and hence, it can be concluded the closed-loop system is FTS. Finally,
	by solving the ordinary differential inequality \eqref{eq_vdot_omega1}, it follows that:
	\begin{eqnarray} \nonumber
	\int_{V(x(t_{k}))}^{V(x(t))} \frac{dV}{V^a} &\leq& \int_{t_k}^{t} -c\lambda dt \\ \nonumber
	\frac{V(x(t))^{1-a}-V(x(t_k))^{1-a}}{1-a} &\leq& -c\lambda(t-t_k)	
	\end{eqnarray}	
	 which leads to the inequality \eqref{eq_V_DoS1}.
\end{proof}

\section{Finite-Time Stability under Denial-of-Service}  \label{sec_DoS}
In this section, the main result on the finite-time stability under DoS attacks is provided which is defined as Problem 2. We first model the DoS, define related concepts and assumptions, and then provide the sampling scheme such that the system remains FTS.

Let $\{\sigma_n\}_{n \in \mathbb{N}_0}$ denote the time sequences of the DoS occurrence and $\{\tau_n\}_{n \in \mathbb{N}_0}, \tau_n>0$ denote the duration of the $n$-th DoS on the input communication. Further, define $H_n =  [\sigma_n,\sigma_n+\tau_n)$ as the $n$-th DoS interval.
Assume that during each DoS time interval the actuator can either, (i) generate an input based on the most recent control signal update, or (ii) generate zero-input.
The case of zero-input strategy for control of linear systems under lossy links is considered in \cite{schenato2009zero}. In this paper we consider case (i). Define $ \Theta(t)$ as the set of all time intervals during which no DoS occurs, i.e. $ \Theta(t)=[0,t]\setminus\bigcup_{n\in \mathbb{N}_0} \{H_n\}$. In the sample-and-hold scenario, we have $u(t)=\psi(x(t_{k(t))})$ where $k(t) = \mbox{sup}\{k \in \mathbb{N}_0|t_k \in \Theta(t)\}$ denotes the last event instant with successful data transmission over the network. In order to characterize DoS, the following assumptions are made over the interval $[0,t)$.

\begin{ass} \label{ass_dos1}
DoS frequency:  For all $t \geq 0$, there exist $\eta \geq 0$ and $\tau_D > 0$ such that,  $n(t) \leq \eta + t/\tau_D$ where $n(t)$ denotes the number of off/on DoS transitions in the interval $[0,t)$.
\end{ass}

\begin{ass} \label{ass_dos2}
DoS duration: For all $t \geq 0$, there exist $\kappa \geq 0$ and $\theta >1$ such that, $|\Xi(t)|\leq \kappa + t/\theta$ where $\Xi(t) = \bigcup_{n\in \mathbb{N}_0} \{H_n\}$ denotes the total interval of DoS over $[0,t)$.
\end{ass}

It should be noted that $\tau_D$ bounds the average \textit{dwell-time} between two consecutive DoS intervals, see \cite{hespanha1999stability} for more information. In fact, $1/\tau_D$ is the upper-bound on the frequency of off/on DoS transitions. Moreover, $1/\theta$ is a measure of the time fraction over which the DoS occurs, and therefore $t/\theta$ can be interpreted as the average DoS duration.

\begin{rem}
	The above assumptions on the  DoS duration/frequency are practically motivated by the fact that there are several techniques to mitigate DoS attacks, e.g. high-pass filter and spreading methods. These mitigation techniques limit the duration time and frequency of DoS intervals over which input communication is denied \cite{de2015input} and justify Assumptions~\ref{ass_dos1} and~\ref{ass_dos2}. As an example, the attack scenario in \cite{gupta2010optimal} considers that out of $\mc{N}$ possible communications $\mc{M}<\mc{N}$ are denied. This is a special case of the assumptions in this paper, where $\kappa = 0$, $\theta = \infty$ in Assumption~\ref{ass_dos2}, and $\eta = \mc{M}$, $ \tau_D = {\delta}\frac{\mc{M}}{\mc{N}}$  with some ${\delta}>0$ in Assumption~\ref{ass_dos1} .
\end{rem}

Considering Assumptions~\ref{ass_lyapunov}-\ref{ass_dos2}, the \textit{hybrid}  event-triggered mechanism is proposed as follows:
\begin{itemize}
	\item Case (I): if $t_k$ is not in a DoS interval, then the next event instant is defined similar to \eqref{eq_sampling1} as,
	\begin{eqnarray} \label{eq_sampling}
	t_{k+1} = \inf\{t>t_k~|~{\gamma}(4\|e(t)\|)>c(1-\lambda)V^a(x(t))\}
	\end{eqnarray}
with $0<\lambda<1$ and $e(t) = x(t_{k(t)})-x(t)$.
 \item Case (II): if $t_k$ is in a DoS interval, then the next event instant is defined such that for $\Delta_k=t_{k+1}-t_k$ we have $\underline{\Delta}\leq \Delta_k\leq\overline{\Delta}$, where $\overline{\Delta}>0$   and $\underline{\Delta}>0$ are the upper-bound and lower-bound of the inter-event interval, respectively.
\end{itemize}

In the following, first it is shown that  the proposed hybrid event-triggering mechanism does not show Zeno behavior.

\begin{lem}
	The  event-triggering mechanism defined by Case~(I)-(II) does not show Zeno behavior almost at every point.
\end{lem}
\begin{proof}
	Note that the sampling during the DoS attack (Case~(II)) is lower-bounded by $\underline{\Delta}>0$. This along with the proof of Lemma~\ref{lem_zeno1} imply that the hybrid event-triggering mechanism is Zeno-free almost everywhere.
\end{proof}

\begin{lem} \label{lem_omega}
	Consider system \eqref{eq_syst} with feedback $u(t)=\psi(x(t_{k(t))}))$, $t\in [t_k,t_{k+1})$  under the event-triggering mechanism  Case (I)-(II). Then, if $t_k$ is in a DoS interval, it follows that:
		\begin{align} \nonumber
		V^{1-a}(x(t)) \leq& V^{1-a}(x(t_{k(t_k)+1})))\\
		&+(1-a)\omega_2(t-t_{k(t_k)+1}),
		\label{eq_V_DoS2}
		\end{align}
		where $V$ is the ISS Lyapunov function defined in Assumption 1, $\omega_2= c(1-\lambda)+2\mu$,  and $t_{k(t_k)+1}<t<t_{k+1}$.\footnote{Note that $t_{k(t_k)+1}$ denotes the first possible sampling time-instant that is included in the DoS time interval occurring after $t_k$.}
\end{lem}

\begin{proof} Having $t_k$ in the DoS interval and following Case (II), $t_{k+1} = t_k + \Delta_k$. Assuming that a successful event instant $k(t_k)$ occurs before $t_k$, it follows that:
	\begin{eqnarray} \nonumber
	\|e(t)\| \leq \frac{1}{4} \gamma^{-1}\big(c(1-\lambda) V^a(x(t)) \big),~ t_{k(t_k)}\leq t \leq t_{k(t_k)+1}.
	\end{eqnarray}
Following the continuity of $V(x(t))$, $x(t)$, and $e(t)$ at $t_{k(t_k)+1}$, it follows that:
	\begin{align} \nonumber
	\|x(t_{k(t_k)})\| \leq& \|x(t_{k(t_k)+1})\| + \|e(t_{k(t_k)+1})\|\\ \nonumber
	 \leq& \|x(t_{k(t_k)+1})\| + \frac{1}{4} \gamma^{-1}\big(c(1-\lambda) V^{a}(x(t_{k(t_k)+1}))\big) \\ \nonumber
	 \leq& \frac{1}{4} \gamma^{-1}(\mu V^{a}(x(t_{k(t_k)+1}))) \\ \label{eq_lem_omega2}
	 &+ \frac{1}{4} \gamma^{-1}\big(c(1-\lambda) V^{a}(x(t_{k(t_k)+1}))\big),
	\end{align}	
where the last inequality is written based on Assumption~\ref{ass_mu} and \eqref{eq_storage}.
Note that for $\gamma$ as a $\mc{K}_\infty$-class function and $a,b \geq 0 $ we have $\gamma(a+b) \leq \gamma(2a)  + \gamma (2b)$. Using this inequality for $e(t) = x(t_{k(t_k)})-x(t)$, $t_{k(t_k)} \leq t \leq t_{k+1}$, it follows that:
	\begin{align*}
	\gamma(\|e(t)\|) &\leq \gamma(2\|x(t_{k(t_k)})\|) + \gamma(2\|x(t)\|) \\
    &\leq  (c(1-\lambda)+\mu) V^{a}(x(t_{k(t_k)+1}))+ \mu V^{a}(x(t)).
    \end{align*}
Consequently, it follows from  \eqref{eq_V3}  that
    \begin{align} \nonumber
      \dot{V}(x(t))&\leq (c(1-\lambda)+\mu) V^{a}(x(t_{k(t_k)+1})) +(\mu-\lambda) V^{a}(x(t))\\
       &\leq \omega_2 \max\{ V^{a}(x(t_{k(t_k)+1}))), V^{a}(x(t)) \},    \label{eq_Vdot_omega}
    \end{align}
for  $t_{k(t_k)+1} \leq t<t_{k+1}$. Finally, in order to solve the differential inequality \eqref{eq_Vdot_omega}, consider the differential equation $\dot{v}(t)=\omega_2 \max\{v^{a}(x(t_{k(t_k)+1}))), v^{a}(x(t)) \}$ over $t_{k(t_k)+1} \leq t<t_{k+1}$ with the
 initial condition $v(t_{k(t_k)+1})=V(t_{k(t_k)+1})$. It follows that
\begin{align*}
v(t)=\big ( v^{1-a}(t_{k(t_k)+1})+\omega_2 (t-t_{k(t_k)+1})     \big )^{\frac{1}{1-a}},
\end{align*}
and using the comparison lemma, we have $V(t)\leq v(t)$ which leads to \eqref{eq_V_DoS2}.
\end{proof}
\begin{rem}
	Lemma~\ref{lem_omega} quantifies the rate of divergence of the Lyapunov function in the sense that, in the presence of DoS and no successful feedback, the Lyapunov function may increase while the growth value is upper-bounded as in \eqref{eq_V_DoS2}.
\end{rem}

Let $\Lambda(t)$ represents the union of the time-intervals over which the Lyapunov function $V(x(t))$ may increase. Further, define $\Lambda^c(t)$ as the complement of  $\Lambda(t)$ over the time-interval $[0,t)$, i.e., $\Lambda^c(t)=[0,t)\setminus \Lambda(t)$. Following Assumptions~\ref{ass_dos1} and~\ref{ass_dos2} on the DoS frequency/duration, it follows that \cite{de2016networked}:
\begin{eqnarray} \label{eq_lambda}
|\Lambda(t)| \leq \kappa + \frac{t}{\theta} + \overline{\Delta}(\eta+\frac{t}{\tau_D}).
\end{eqnarray}
Following Theorem \ref{thm_main} and Lemma~\ref{lem_omega}, it follows from  \eqref{eq_V_DoS1} and \eqref{eq_V_DoS2} that
\begin{eqnarray} \nonumber
		V^{1-a}(x(t)) &\leq&  V_0^{1-a}\\ &+&(1-a)(\omega_2|\Lambda(t)|-\omega_1|\Lambda^c(t)|).
		\label{eq_V_lambda}
\end{eqnarray}
In the above, the DoS-related term may increase the Lyapunov function (causing instability), while the other term decreases the Lyapunov function (causing stability). The goal is to determine the conditions on the DoS  frequency/duration such that the stabilizing term is predominant and, therefore, the closed-loop system remains  FTS.

\begin{theorem} \label{thm_main}
	Consider system~\eqref{eq_syst} with input feedback  $u(t)=\psi(x(t_{k(t)}))$ under event-triggered mechanism Case (I)-(II) and Assumptions~\ref{ass_mu}, \ref{ass_dos1}, and~\ref{ass_dos2}. The system is FTS, if
\begin{eqnarray}
\frac{1}{\theta}+\frac{\overline{\Delta}}{\tau_D}<\frac{c\lambda}{c+2\mu}.
\label{eq_main}
\end{eqnarray}
   Furthermore,
\begin{eqnarray}
   \|x(t)\| \leq \alpha_1^{-1}\left(\left(\alpha_2(\|x_0\|)^{1-a}+(1-a)(\rho-\xi t) \right)^{\frac{1}{1-a}} \right),
   \label{eq_main2}
\end{eqnarray}
where,
\begin{eqnarray}	\nonumber
\xi &:=& c\lambda-\left(\frac{1}{\theta}+\frac{\overline{\Delta}}{\tau_D}\right)(c+2\mu),\\ \nonumber
\rho &:=& (c+2\mu)(\kappa+\overline{\Delta}\eta).
\end{eqnarray}
\end{theorem}	
\begin{proof}
First,  it follows from \eqref{eq_lambda} that:
\begin{align*}	
 	\omega_2|\Lambda(t)|&-\omega_1|\Lambda^c(t)| = (c+2\mu)(\kappa+\overline{\Delta}\eta) \\
	 &+ t\left(\left(\frac{1}{\theta}+\frac{\overline{\Delta}}{\tau_D}\right)(c+2\mu)-c\lambda\right)= \rho-\xi t.
\end{align*}	
Then, based on the condition \eqref{eq_main}, we have $\xi>0$ and it follows from \eqref{eq_V_lambda} that
\begin{align*}
V^{1-a}(x(t)) &\leq  V_0^{1-a}+(1-a)(\rho-\xi t),
\end{align*}
which guarantees finite-time stability. Moreover, the settling-time is upper-bounded as:
\begin{eqnarray} \nonumber
T(x_0) &\leq&  \frac{V_0^{1-a}+(1-a) \rho }{\xi}
\end{eqnarray}
The second part of the theorem follows from equations \eqref{eq_storage}, \eqref{eq_lambda}, and \eqref{eq_V_lambda} as,
\begin{eqnarray} \nonumber
\alpha_1(\|x\|) &\leq& V(x(t)), \\ \nonumber
V^{1-a}(x(t)) &\leq& V_0^{1-a}+(1-a)(\rho-\xi t), \\ \nonumber
V_0 &\leq& \alpha_2(\|x(0)\|),
\end{eqnarray}
and we have,
\begin{eqnarray} \nonumber
\alpha_1^{1-a}(\|x(t)\|) \leq  \alpha_2^{1-a}(\|x(0)\|)+(1-a)(\rho-\xi t),
\end{eqnarray}
which leads to \eqref{eq_main2}.
\end{proof}

\section{Simulation Study}\label{sec_example}
Consider a nonlinear system in the form:
\begin{eqnarray} \label{eq_example_system}
\dot{x}(t)=-\mbox{sgn}(x(t))|x(t)|^{\frac{1}{2}}+x(t)+u(t),
\end{eqnarray}
with the event-triggered control input $u(t)=\psi(x(t_{k(t))})=-2x(t_{k(t)})=-2(x(t)+e(t))$  and  the Lyapunov function  as $V(x(t))=x^2(t)$. Then, it follows that
\begin{align} \nonumber
\dot{V}(x(t),e(t))&=2x(t)(-\mbox{sgn}(x(t))|x(t)|^{\frac{1}{2}}-2x(t)-2e(t))\\&=-2|x(t)|^{\frac{3}{2}}-2x^2(t)-4x(t)e(t).\label{eq_Vdot}
\end{align}
Next, by using the Young's inequality for the term $x(t)e(t)$ it follows that:
\begin{eqnarray} \nonumber
\dot{V}(x(t),e(t))&\leq& -2|x(t)|^{\frac{3}{2}}-2x^2(t)+2x^2(t)+2e^2(t)\\ \nonumber
&\leq& -2|x(t)|^{\frac{3}{2}}+2e^2(t)\\ \nonumber
&\leq& -2V^{\frac{3}{4}}(t)+2e^2(t).
\end{eqnarray}
Following \eqref{eq_V3},  we have $c=2$, $a=\frac{3}{4}$, and $\gamma(r)=2r^2$ which implies that the system is FTISS from Lemma~\ref{lem_lyapunov_order}. Hence, by selecting $\lambda=0.5$, the event-triggered condition \eqref{eq_sampling1} is given as:
	\begin{align}
	t_{k+1} = \inf\{t>t_k~|~32e^2(t)>|x(t)|^{1.5}\}.	\label{Event_example}
\end{align}	
Figure \ref{no_dos} shows the state trajectory of system \eqref{eq_example_system} with the proposed event-triggered mechanism \eqref{Event_example}. As shown in this figure, the state $x(t)$ converges to zero in a finite-time and the total number events in this scenario is 40. Moreover, no data is sent through the network after $t=2.03$ which is due to the finite-time convergence property of the closed-loop system.  In comparison with the time-triggered network communication with the sampling time $T=0.02$, the number of data exchange in 5 seconds is reduced from 250 samples in the time-triggered  fashion to 40 ones in the proposed event-triggered scheme which shows a significant reduction in the communication between the plant and the controller.

\begin{figure}
	\centering
	\includegraphics[width=3in,height=1.6in]{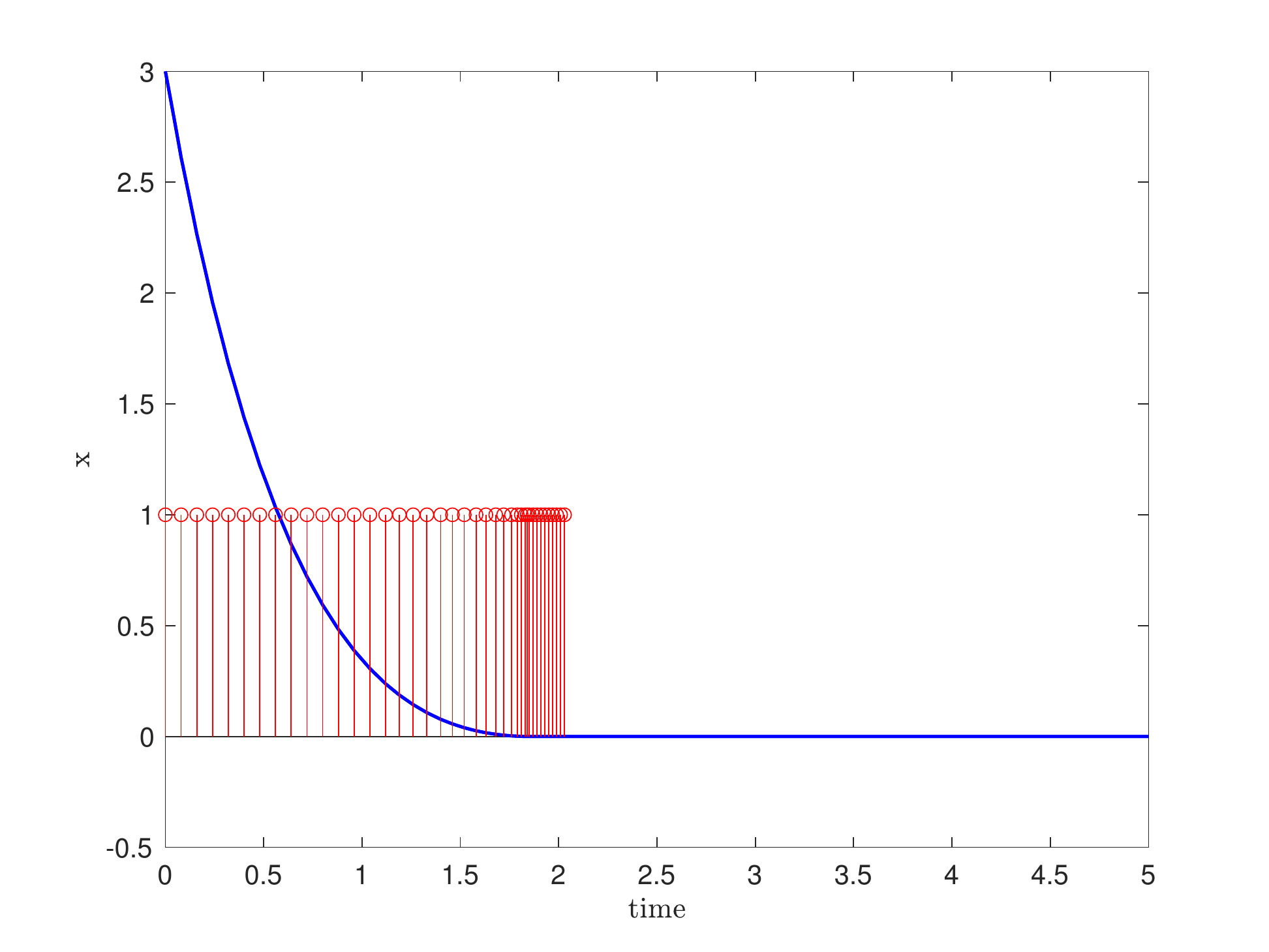}
	\caption{The state trajectory corresponding to $x(0)=3$ without DoS attack. Solid blue line: state evolution, Red bars: the event-time instant $t_k$. }
	\label{no_dos}
\end{figure}

 Next, the performance of the proposed event-triggering scheme in the presence of DoS attack is demonstrated. It follows from  \eqref{eq_storage} that $\alpha_1(r)=r^2$ and $\alpha_2(r)=3r^2$.
Therefore, for $|x|<3$, any $\mu\geq55.43$ satisfies Assumption~\ref{ass_mu} and the hybrid event-triggered mechanism \eqref{Event_example}  with $\Delta_k=\overline{\Delta}=0.1$ is selected.  Using these parameters in Theorem~\ref{thm_main}, one can find an upper-bound on frequency/duration of DoS intervals under which the system is guaranteed to remain FTS as follows:
\begin{eqnarray} \label{eq_bound}
\frac{1}{\theta}+\frac{0.1}{\tau_D}<\frac{1}{112.86} \approx 0.00886.
\end{eqnarray}
It should be noted that this bound is conservative and can in practice be larger than the theoretical one. The same observation is reported in  \cite{de2016networked} and  this is mainly due to the fact that 
 the condition in Assumption~\ref{ass_mu} bounds the derivative of the Lyapunov function irrespective of the specific nonlinear system. In other words, the bound in \eqref{eq_bound} holds for any nonlinear system for which the Lyapunov function satisfies similar condition as in Assumption~\ref{ass_mu}.

\begin{rem}
	The bound in \eqref{eq_bound} depends on the  parameters $\mu$, $\lambda$, and $\Delta_k$ which can be determined by the designer. This implies that  the designer can potentially tune the parameters to manage, for example, the convergence performance or the communication rate based on the existing resources. Further, the control input $u(t)$ can be designed for different purposes, for example, for robustness against disturbances. These provide desirable design flexibility for several implementation options.
\end{rem}

\begin{figure}
	\centering
	\includegraphics[width=3.2in,height=1.6in]{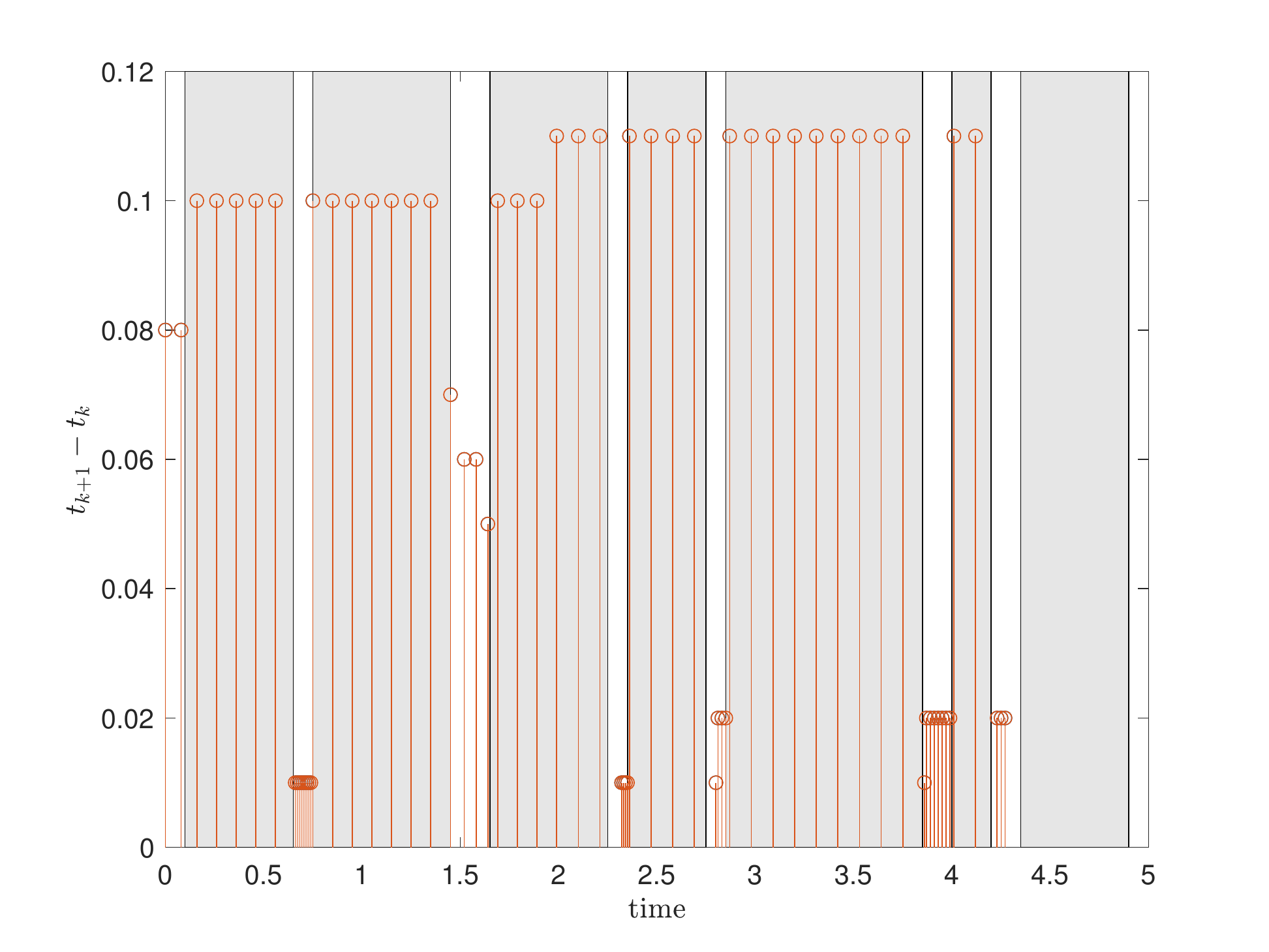}
	\caption{ The inter-event interval $t_{k+1}-t_k$  in the presence of DoS attacks. The vertical gray stripes represent the DoS time-intervals.}
	\label{fig_delta}
\end{figure}

\begin{figure}
	\centering
	\includegraphics[width=3.2in,height=1.6in]{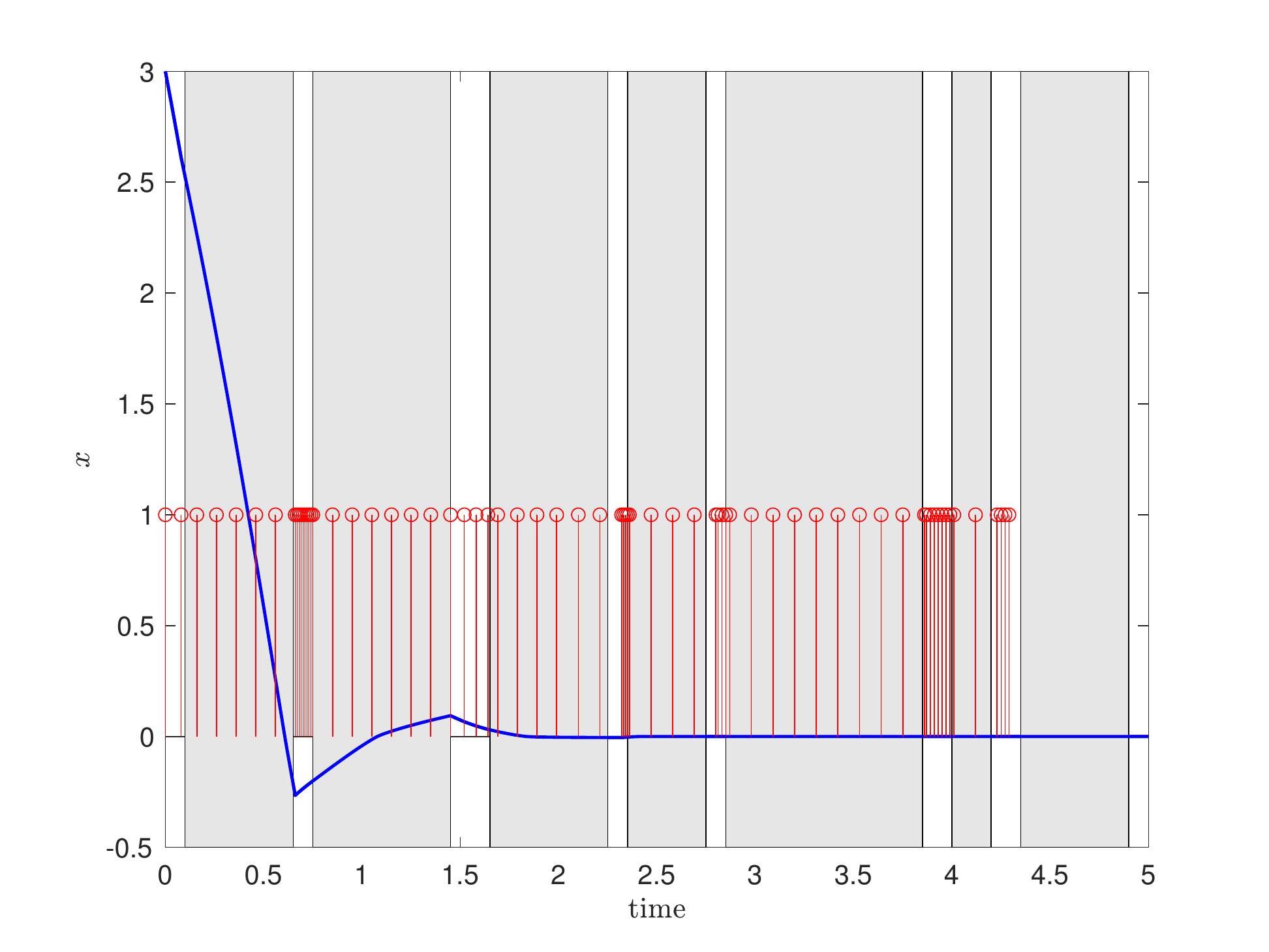}
	\caption{The state trajectory corresponding to $x(0)=3$ in the presence of  DoS attacks. Solid blue line: state evolution, Red bars: the event-time instant $t_k$. The vertical gray stripes in the background represent the DoS time-intervals. }
	\label{fig_x_dos}
\end{figure}

\begin{figure}
	\centering
	\includegraphics[width=3.2in,height=1.6in]{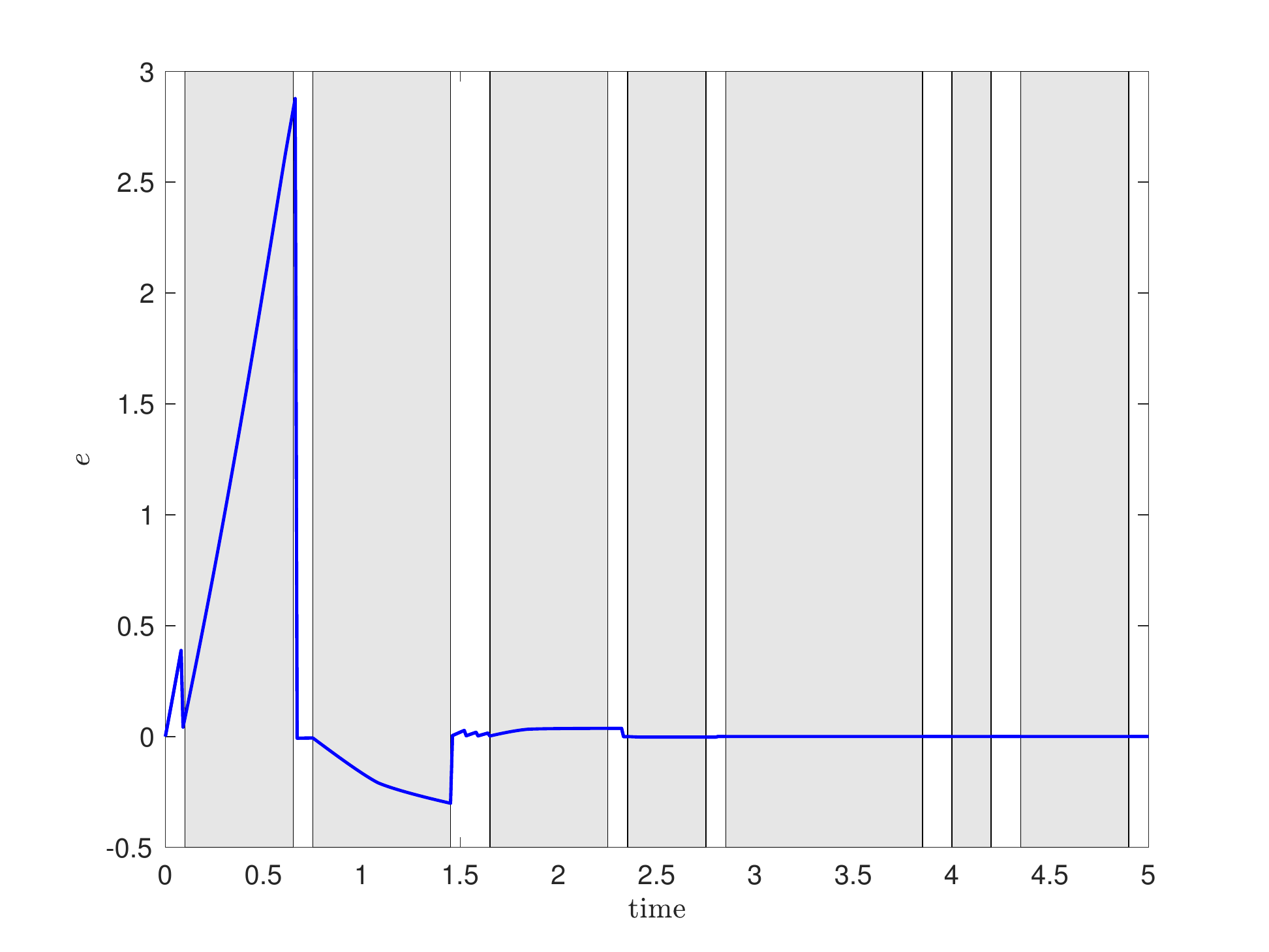}
	\caption{The error trajectory $e(t)$ corresponding to $x(0)=3$ in the presence of  DoS attacks. The vertical gray stripes represent the DoS time-intervals.}
	\label{fig_e_dos}
\end{figure}

For numerical simulation, as shown Fig. \ref{fig_delta}, we randomly generate DoS attack with the frequency and duration of $n(5)=7$ and $\Xi(5)\simeq 4$, respectively, in the time-interval $[0,5]$.
This figure particularly represents the outcome of the proposed hybrid event-triggered mechansim, where  during the off-periods of DoS (Case (I)), the event-triggered condition \eqref{Event_example} is used and  during the  DoS attack, we have ${\Delta}_k=0.1$. The average duty cycle of DoS signal is approximately $80\%$, implying $80\%$ denial of input transmissions over time (in average). Figures \ref{fig_x_dos} and \ref{fig_e_dos} show the state trajectory $x(t)$ and the corresponding error trajectory $e(t)$ of the proposed hybrid event-triggered mechanism. As shown in Figure \ref{fig_x_dos}, the state of the system converges to zero in a finite-time and the total number events in this scenario is 68 which as expected is more than the previous case corresponding to no DoS attack. However, even in the presence of DoS attack, the number of event is much less than the time-triggered communication mechanism. Figure \ref{fig_e_dos} shows the time evolution of the error signal $e(t)$ and as expected during the DoS attack, due to the denial of data exchange, the error signal  increases while after the removal of DoS, it returns back to zero.

\begin{figure}
	\centering
	\includegraphics[width=3.2in,height=1.6in]{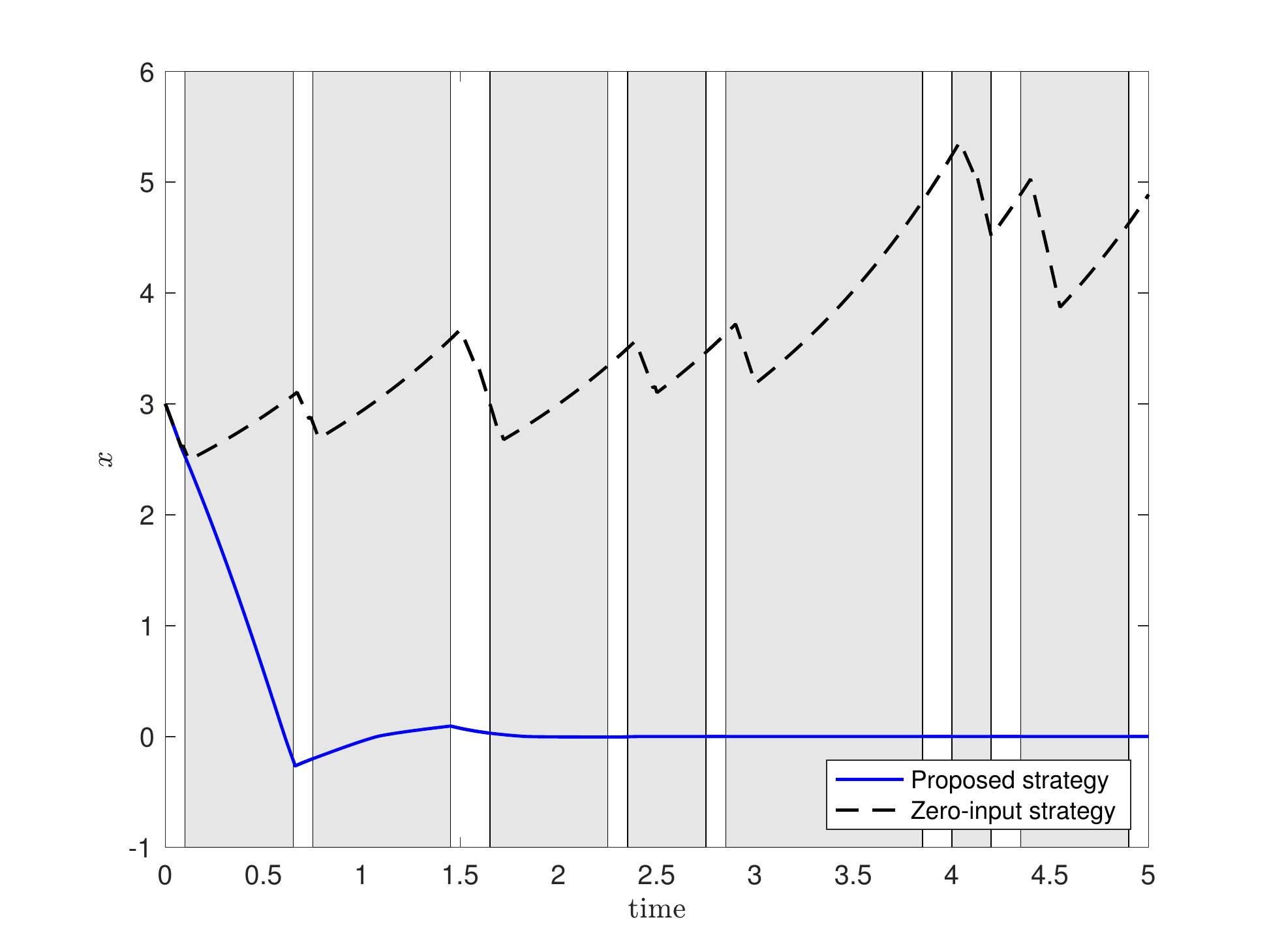}
	\caption{The state trajectory corresponding to $x(0)=3$ in the presence of  DoS attacks (represented by gray stripes). Blue solid line: proposed approach, Black dashed line:  zero-input strategy \cite{schenato2009zero}. }
	\label{fig_x_comparison}
\end{figure}
For comparison,  the zero-input strategy proposed in \cite{schenato2009zero} is also simulated and the state trajectory of our proposed approach as well as the one in \cite{schenato2009zero} are shown in Fig. \ref{fig_x_comparison}. Following the proposed strategy, during the DoS intervals, the input is constant and equal to the most recent control signal updated during no-DoS interval. A different strategy is proposed in \cite{schenato2009zero} by considering zero-input during the DoS intervals.
As it can be seen from Fig.~\ref{fig_x_comparison}, the proposed event-triggered mechanism under DoS (with  input transmission denial over $80\%$ of the simulation time) successfully stabilizes the  system~\eqref{eq_example_system} in a finite-time and it outperforms the zero-input strategy \cite{schenato2009zero} which is not necessarily stable as shown in Fig.~\ref{fig_x_comparison}. As the final comment, note that the problem of the FTS and FTISS systems under DoS is to great extent unexplored in the literature and there is no other paper specifically discussing this topic for the sake of comparison.

\section{Concluding Remarks}\label{sec_con}
In this paper, finite-time stabilizing control under denial-of-service attack on input transmissions is considered. The interest in FTS systems  has  recently been increased due to their decreased response time, for example, in robotic applications.
The practical results given on FTISS and FTS systems and related Lyapunov analysis alleviate the technicality for non-smooth feedback design. We particularly relate the duration/frequency of DoS attack to finite-time stability where  no assumption on the information available to the DoS attacker regarding the sampling logic, underlying nonlinear system, and control input is considered.
We propose an event-based sampling logic which is flexible in terms of design parameters, allowing to account for, e.g., limitations in communication resources. Note that various control methods may be applied for finite-time stability of the underlying system as our resilient sampling does not impose any constraint on the input.

It should be mentioned that, during DoS attack one can either apply the zero-input or the last updated input to the system. As discussed in the simulation section, using the last updated input outperforms the zero-input strategy as zero-input is generally destabilizing for general open-loop unstable systems. Even if the attacker is intelligent, since the underlying system is generally unstable, zero input has no better effect than using the last updated input. Note that, during DoS the controller cannot update the input therefore the zero input implies complying with the attacker which in turn result in instability. Therefore, during the DoS interval there is no better strategy than using the last updated input signal. In general, the optimality criteria for the superiority of the either of the two strategies only can be defined when the probability of DoS attacks (or lossy links, packet loss,  failed data) is known \cite{schenato2009zero}. However, in this paper as discussed in the introduction it is assumed that the DoS attacks do not follow a specific probability distribution. Therefore, no optimality criteria can be considered as the DoS attacks are unpredictable. One may consider a game-theoretic approach for the cyber defender if the intelligent attacker follows a specific pattern (or game) and this can be considered as one of the future research directions.

Future research are considered in the following directions. Self-triggering sampling methods \cite{anta2010sample} based on the prediction of system state could be applied, for example, in case of asynchronous denial of measurement and input channels. Extension to \textit{distributed} control application  \cite{de2013robust} is another promising direction of  research and finally
finite-time control methods based on the sign function and communicating single-bit of information  \cite{scientia,single_bit} under DoS is another interesting application.

\bibliographystyle{IEEEtran}
\bibliography{bibliography}

\end{document}